\documentclass[a4paper,UKenglish,cleveref, autoref, thm-restate]{lipics-v2021}

\usepackage{graphicx}
\usepackage{amsmath}
\usepackage{algorithm} 
\usepackage{algpseudocode} 
\usepackage{dsfont}
\usepackage{todonotes}
\usepackage{verbatim}
\usepackage{thmtools}
\usepackage{svg}
\usepackage{subcaption}

\usepackage{amsthm}
\usepackage{amssymb}

\newcommand{\comm}{\textsf{Comm}}
\newcommand{\polylog}{\mathrm{polylog}}

\newcommand{\byzantineconnectivity}{\textsc{Byzantine-Connectivity}}
\newcommand{\byzantinereco}[1]{\textsc{Byzantine Recognition of } #1}

\newcommand{\reco}[1]{\textsc{Recognition-of } #1}

\newcommand{\cO}{\mathcal{O}}
\algblockdefx{INP}{ENDINP}[1]%
  {\textbf{In Parallel }#1 \textbf{do}}%
  {\textbf{End Parallel}}

\algblockdefx{IF}{ENDIF}[1]%
  {\textbf{if}#1 \textbf{then}}%
  {\textbf{end if}}

\algblockdefx{FOR}{ENDFOR}[1]%
  {\textbf{for}#1 \textbf{do}}%
  {\textbf{end for}}

\pdfoutput=1 
\hideLIPIcs  

\title{Recognizing Hereditary Properties in the Presence of Byzantine Nodes} 

\author{David Cifuentes-Núñez}{Departamento de Ingeniería Matemática, Universidad de Chile, Santiago, Chile}{dcifuentes@dim.uchile.cl}{https://orcid.org/0009-0002-5464-0372}{}

\author{Pedro Montealegre}{Facultad de Ingenier\'ia y Ciencias, Universidad Adolfo Ib\'a\~nez, Santiago, Chile}{p.montealegre@uai.cl}{https://orcid.org/0000-0002-2508-5907}{This work
was supported by ECOS-ANID  ECOS240020, STIC-AMSUD ECODIST AMSUD240005,  Centro de Modelamiento Matemático (CMM), FB210005, BASAL funds
for centers of excellence from ANID-Chile, FONDECYT 1230599 and ANID-MILENIO-NCN2024\_103}

\author{Ivan Rapaport}{Departamento de Ingeniería Matemática, Universidad de Chile, Santiago, Chile \and Centro de Modelamiento Matemático (UMI 2807 CNRS), Universidad de Chile, Santiago, Chile}{rapaport@dim.uchile.cl}{https://orcid.org/0000-0002-2969-5083}{This work
was supported by  Centro de Modelamiento Matemático (CMM), FB210005, BASAL funds
for centers of excellence from ANID-Chile and FONDECYT 1220142}

\authorrunning{D. Cifuentes-Núñez, P. Montealegre and I. Rapaport} 

\Copyright{David Cifuentes-Núñez, Pedro Montealegre and Ivan Rapaport} 

\ccsdesc[500]{Theory of computation~Distributed algorithms}

\keywords{Byzantine protocols, congested clique, hereditary properties} 

\nolinenumbers 

\EventEditors{Andrei Arusoaie, Emanuel Onica, Michael Spear, and Sara Tucci-Piergiovanni}
\EventNoEds{4}
\EventLongTitle{29th International Conference on Principles of Distributed Systems (OPODIS 2025)}
\EventShortTitle{OPODIS 2025}
\EventAcronym{OPODIS}
\EventYear{2025}
\EventDate{December 3--5, 2025}
\EventLocation{Ia\c{s}i, Romania}
\EventLogo{}
\SeriesVolume{361}
\ArticleNo{26}

\begin{document}

\maketitle

\begin{abstract}
Augustine et al. [DISC 2022] initiated the study of distributed graph algorithms in the presence of Byzantine nodes in the congested clique model. In this model, there is a set $B$ of Byzantine nodes, where $|B|$ is less than a third of the total number of nodes. These nodes have complete knowledge of the network and the state of other nodes, and they conspire to alter the output of the system. The authors addressed the connectivity problem, showing that it is solvable under the promise that either the subgraph induced by the honest nodes is connected, or the graph has $2|B|+1$ connected components. 

In the current work, we continue the study of the Byzantine congested clique model by considering the recognition of other graph properties, specifically \emph{hereditary properties}. A graph property is \emph{hereditary} if it is closed under taking induced subgraphs. Examples of hereditary properties include acyclicity, bipartiteness, planarity, and bounded (chromatic, independence) number, etc.

For each class of graphs $\mathcal{G}$ satisfying  a hereditary property (a hereditary graph-class), we propose a randomized algorithm which, with high probability, (1) accepts if the input graph $G$ belongs to $\mathcal{G}$, and (2) rejects if $G$ contains at least $|B| + 1$ disjoint subgraphs not belonging to $\mathcal{G}$. The round complexity of our algorithm is
$$\cO\left(\left(\dfrac{\log \left(\left|\mathcal{G}_n\right|\right)}{n} +|B|\right)\cdot\textrm{polylog}(n)\right),$$ where \(\mathcal{G}_n\) is the set of \(n\)-node graphs in \(\mathcal{G}\). 

Finally, we obtain an impossibility result that proves that our result is tight. Indeed, we consider the hereditary class of acyclic graphs, and we prove that there is no algorithm that can distinguish between a graph being acyclic and a graph having $|B|$ disjoint cycles.
\end{abstract}
\newpage
\section{Introduction}

The congested clique model has been intensively studied over the last decade \cite{censor2015algebraic,drucker2014power,korhonen2018towards}. The model is defined as follows. Given a graph $G =(V, E)$, initially
each node, uniquely identified by an ID in $\{1, \dots, n\}=[n]$, knows only its neighborhood in $G$. The nodes want to collectively decide some predicate of \(G\). For instance, decide whether \(G\) is connected.  Communication happens in
synchronous rounds over a clique (the complete communication network where $G$ is embedded), and each
node can only send $O(\log n)$ bits to each other node per round. The goal is to design  distributed algorithms such that all nodes end up knowing the correct answer regarding the input graph $G$. The congested clique  model is extremely powerful, and fast algorithms of $O(\log n)$ or even $O(1)$ rounds have been devised to solve a range of problems such as coloring \cite{chang2019complexity}, MIS \cite{ghaffari2017distributed}, MST \cite{jurdzinski2018mst}, etc.

\subsection{The Byzantine Congested Clique Model}

In \cite{augustine2022byzantine} Augustine et al. initiated the study of distributed graph algorithms {\emph{under the presence of Byzantine nodes}} in the congested clique model. A Byzantine node, in contrast with an honest node,  is a malicious node with unlimited computational power and full knowledge of the states of the honest nodes, that can behave arbitrarily and collude with other Byzantine nodes. In the \emph{Byzantine congested clique model}, there is a set \(B\) of Byzantine nodes, which represent less than one third of the total number of vertices. At the begining of an algorithm, an adversary chooses the vertices that are Byzantine, and uses them to fool the honest nodes making them decide a wrong output. The honest nodes know the number of Byzantine nodes but not their identities. The adversary has full knowledge of the network, unlimited computational power, and knows the algorithm executed by the honest nodes. The objective is to devise algorithms that can decide properties of the input graph $G$ within this highly adversarial framework.

A first issue to consider is how to define the (distributed) output of the model. In classical local decision, we ask that, in the yes-instances, every node accepts, while in the no-instances at least one node rejects. This way of deciding is not suitable under the presence of Byzantine nodes, as these nodes could, for instance, always reject on yes-instances. To cope with this issue, we consider a unanimous decision rule over the subset of honest nodes: in yes-instances, every honest node must accept, and in no-instances all honest nodes must reject. 

Unfortunately, the behavior of the Byzantine nodes could make impossible to distinguish between yes and no-instances. Consider for example the problem of deciding whether a graph is connected, and suppose we are given a disconnected input graph with two connected components. A pair of Byzantine nodes in different connected components could fool any algorithm by faking an edge between them, and otherwise act honestly. A similar situation happens when a connected graph has a bridge and both endpoints of the bridge are Byzantine nodes.

The previous observation was the reason why Augustine et al. \cite{augustine2022byzantine} consider promise problems, the situation where the input is promised to belong to a particular subset. More precisely, their approach was to distinguish the case where the graph induced by the subset of honest nodes is a yes-instance, from the case where the whole graph, including the Byzantine nodes, is a no-instance that is \emph{far} from being a yes-instance. The precise definition of \emph{farness} depends on the particular decision problem we are considering.

\subsection{Byzantine Connectivity}

Let $f \geq 1$. In the problem \(\byzantineconnectivity(f)\), introduced in \cite{augustine2022byzantine},
every honest node must accept when the graph induced by the honest nodes is connected, and every honest node must reject when the graph is \(f\)-\emph{far} from being connected. A graph is \(f\)-far from being connected if it contains at least \(f+1\) connected components. In the intermediate case (i.e., when the input graph is not connected and has at most \(f\) connected components), the output of the honest nodes could be accept or reject, but the decision must be unanimous.

The main results of \cite{augustine2022byzantine} show that \(\byzantineconnectivity(2|B|)\) can be solved in the Byzantine congested clique model (recall that $B$ is the set of Byzantine nodes). More precisely, the problem can be solved with high probability by a randomized algorithm running in $\cO(\polylog(n))$ rounds, or by a deterministic algorithm running in \(\cO(n)\) rounds. The authors also show that no randomized algorithm  can solve  \(\byzantineconnectivity(f)\) with probability greater than \(1/2\) unless \(f \geq 2|B|\). These algorithms are devised under the condition that 
the fraction of Byzantine nodes is less than \(1/3\). In this work we assume the same constraint on  $|B|$.

\subsection{Hereditary Properties} 

Our goal is to design algorithms in the Byzantine congested clique model for the broad
spectrum of hereditary properties. Note that a property \(\mathcal{P}\) 
defines a class of graphs $\mathcal G$,  where $G$ satisfies $\mathcal P$ iff 
$G\in \mathcal G$. A class of graphs $\mathcal G$ is hereditary 
if every induced subgraph of a graph  $G \in \mathcal G$ is also in $\mathcal{G}$.
In other words, the class is closed under induced subgraphs. Examples of hereditary properties (and hereditary graph classes) are planarity, acyclicity, bounded chromatic number, bipartiteness, etc. Given a hereditary graph class \(\mathcal{G}\), we study the problem of deciding whether a given graph belongs to \(\mathcal{G}\), also called the \emph{recognition of} \(\mathcal{G}\). 

Byzantine nodes can collude to disseminate misinformation about their neighbors and eventually
deceive the honest nodes into producing an incorrect output. For instance, consider 
the problem of deciding bipartiteness and the case where the input graph is indeed bipartite. The adversary could pick two nodes in the same partition  in order to fake an edge between them, and otherwise act as an honest node. Since the honest nodes cannot distinguish if a node is Byzantine or not, they cannot know which edges are fake, and hence cannot distinguish a yes from a no instance. In other words, as in the case of connectivity, we  have to consider promise problems.

Let \(f\geq 1\) and let \(\mathcal{G}\) be a hereditary graph class. We say that a graph \(G\) is \emph{$f$-far from $\mathcal{G}$} if $G$ contains at least $f+1$ disjoint induced subgraphs not in $\mathcal{G}$. Then, we define the problem \(\byzantinereco{\mathcal{G}(f)}\) as the following promise problem:

\bigskip

\noindent {\bf Input:} A graph \(G=(V,E)\) and a set \(B\) of Byzantine nodes. \\

\noindent {\bf Task:} 
\begin{itemize}
\item If \(G\) belongs to \(\mathcal{G}\), every honest node must accept. 
\item If \(G\) is \(f\)-far from $\mathcal{G}$, every honest node must reject.
\item Otherwise, each node can accept or reject.
\end{itemize}

The parameter $f > 0$  is referred to as the \emph{gap of the promise problem}. For simplicity, when the input graph is not in $\mathcal G$  and is not $f$-far 
from $\mathcal G$, we do not require unanimous output from the honest nodes, i.e., some nodes can accept while others can reject. This differs from the definition of  Augustine et al. in \cite{augustine2022byzantine}, where the honest nodes are mandated to produce the same output on any given instance. Nevertheless, this difference does not impact our results, as achieving unanimity among the honest nodes can be achieved by using a \emph{Byzantine everywhere agreement} algorithm (refer to Theorem 10 and Algorithm 1 in \cite{augustine2022byzantine}).

\subsection{Our Results}

We show that for every hereditary graph class 
\(\mathcal{G}\)  there is an algorithm that solves \\
\(\byzantinereco{\mathcal{G}(|B|)}\), as long as the fraction of Byzantine nodes is smaller than 1/3 (we omit this  hereafter, because this assumption is part of the model). More formally,  we show the following theorem:

\begin{restatable}{theorem}{main}\label{theo:main}
Let  \(\mathcal{G}\) be a hereditary graph class. There is a randomized algorithm in the Byzantine congested clique model that   solves   \(\byzantinereco{\mathcal{G}(|B|)}\) with high probability   in 
$$\cO\left(\left(\dfrac{\log \left(\left|\mathcal{G}_n\right|\right)}{n} +|B|\right)\cdot\polylog(n)\right)$$ rounds where \(\mathcal{G}_n\) is the set of \(n\)-node graphs in \(\mathcal{G}\). 
\end{restatable}

  The logarithm of the cardinality of \(\mathcal{G}_n\) is called the \emph{growth} of \(\mathcal{G}\). 
For many classes of graphs, we can bound their growth by an explicit sub-quadratic function. For instance,  forests, planar graphs, interval graphs, graphs excluding fixed minor have \emph{factorial growth}, meaning that \(\log \left(\left|\mathcal{G}_n\right|\right)\) is \( \cO(n \log n)\). For  all those classes, our algorithm runs in $\cO\left(|B|\cdot \polylog(n)\right)$ rounds. 

Finally, we obtain a negative result, which show the tightness of the gap of the promise problem. More precisely, we prove that  when \(f<|B|\), it is impossible to solve 
{\sc Byzantine-Recognition-of-Forests}\((f)\) in the Byzantine Congested Clique model. This holds even for randomized algorithms. 

\begin{restatable}{theorem}{secondmain}\label{theo:impossibility}
If $f < |B|$ then no randomized algorithm in the Byzantine congested clique model can solve  \(\textsc{Byzantine-Recognition-of-Forests}(f)\) with an error probability strictly smaller than \(1/2\).
\end{restatable}

\subsection{Related Work}

The notion of Byzantine nodes has a long history, originating from the famous papers by Lamport, Pease, and Shostak \cite{lamport2019byzantine,pease1980reaching}. The algorithms developed since then have mostly focused on primitive tasks such as consensus, leader election, broadcast, etc. \cite{augustine20,ben2006byzantine,feldman1997optimal,king2011breaking,king2006scalable}, but the result in \cite{augustine2022byzantine} constitutes  a first step towards the study of algorithms that solve {\emph{more complex problems}} under the presence of (a large number of) Byzantine nodes.

The congested clique model is a message-passing paradigm in distributed computing, introduced by Lotker, Patt-Shamir, Pavlov, and Peleg~\cite{lotker2005minimum}. In the congested clique, the underlying communication network forms a complete graph. Consequently, in this network with  diameter 1, the issue of locality  is taken out of the
picture. The primary focus shifts to congestion, a pivotal concern alongside locality in distributed computing. The point is the following: when the communication network is a complete graph, and the cost of local computation is disregarded, the sole impediment to executing any task arises from congestion. The power of the congested clique model comes from the fact that nodes are fully interconnected. Such power explains the ability to simulate efficiently well-studied parallel circuit models \cite{drucker2014power}, implying that establishing lower bounds in the congested clique model  would imply a significant progress in Boolean circuit theory. This is why the quest for lower bounds is considered elusive, to say the least.  As mentioned in the introduction, the congested clique model has been used 
extensively to solves problems such as coloring \cite{chang2019complexity,czumaj2020simple}, MIS \cite{ghaffari2017distributed,konrad2018mis}, algebraic methods \cite{censor2015algebraic}, reconstruction \cite{montealegre2020graph}, and more.

\section{Preliminaries}\label{sec:prelim}

\subsection{Committee of Leaders}

We begin by defining a \emph{committee of leaders} which is a set of nodes containing at least a \(2/3\)-fraction of honest nodes. In our algorithms the committee of leaders will make possible to take global decisions, as well as executing certain subroutines.

The problem of creating a committee of leaders is addressed by Agustine et al. in \cite{augustine2022byzantine}, who use the following result of \cite{king2006scalable} in order to compute a committee of leaders in the Byzantine congested clique model. 

\begin{proposition}\cite{king2006scalable}\label{L1} There is a randomized protocol that w.h.p. computes a committee of leaders $L$ of size $\Theta(\log(n))$ in $\mathcal{O}(\polylog(n))$ rounds, for which at least a fraction of $2/3$ of the committee members are honest nodes. By computing \(L\), we mean that all honest nodes in the network know the identities of the nodes in~$L$.  
\end{proposition}  

\subsection{Committee Structure of a Graph} 

For a given node \( u \), we define a set of nodes as the \emph{committee of \(u\)}, which is used to represent \( u \) in a reliable way. The committees are used to simulate the behavior of the set of nodes within an algorithm designed for the congested clique model (without Byzantine nodes). To achieve this, we require that the majority of the nodes in each committee be honest. This way, we can ensure an ``honest'' execution of an algorithm simply by adopting the majority behavior in each committee. This approach presents two challenges.

First, we need that the honest nodes in the committee of \( u \) learn the neighborhood of  \( u \), in order to accurately simulate \( u \) within an algorithm. This requires \( u \) to communicate its set of neighbors to its committee. A potential complication arises if \( u \) is a Byzantine node, as it might provide a false set of neighbors. Furthermore, \( u \) may mislead different nodes within its committee in different ways.

The second challenge is related to the committee size. If we are interested in keeping the complexity (number of rounds) of an algorithm low, we need to keep the committee sizes small. Given that the neighborhood of a node could contain an arbitrary number of nodes, we  need to execute load balancing algorithms in order to replicate the information within a committee in a reasonable number of rounds. Additionally, we need that a node participates in a bounded number of committees of other nodes, to distribute the workload relatively evenly.

Given a graph $G=(V,E)$ embedded in an $n$-node Byzantine congested clique model, the  \emph{committee} of a node $v\in V$, denoted $\comm(v)$, is a subset of nodes satisfying the following properties:

\begin{enumerate}
    \item The majority of its members are honest. 
    \item Every $u\in \comm(v)$ has an information (not necessary correct) about $N_G(v)$ (the neighborhood of $v$ in $G$). More precisely, for each \(w\in V\), node \(u \in \comm(v)\) owns a bit $c_v^u(w)$ called \emph{consistency bit of \(v\) with respect to \(w\)}.  The value of \(c_v^u(w)\) equals \(1\) if and only if  \(u\) believes that  $vw \in N_G(v)$.
    \item Every honest node in $\comm(v)$ has the same information of $N_G(v)$. Formally, $\forall u_1,u_2\in \comm(v)\setminus B, \forall w\in V, c_v^{u_1}(w)=c_v^{u_2}(w)$. 
\end{enumerate}

Given a committee \(\comm(v)\) of \(v\), we have that all the honest nodes share the same consistency bits. Therefore, in the following we denote by $c_v(w)$ the consistency bit of \(v\) with respect to \(w\) of any of the honest nodes in $\comm(v)$. For each \(v\in V\) we denote \(\comm^{-1}(v)\) the set of nodes \(u\) such that \(v\in \comm(u)\). 
We refer to a \emph{Committee structure} to a collection of committees  \(\comm(G) = \{\comm(v)\}_{v\in V}\) satisfying the following conditions:

\begin{enumerate}
    \item For every pair of nodes \(u,v\in V\setminus B\), node \(u\) knows the members of \(\comm(v)\). 
    \item $\forall u,v \in V\setminus B$, $\forall w \in \comm(u)\setminus B$, node $w$ knows $c_{v}(u)$. 
    \item If $u$ is honest and $uv \notin E$, then $c_u(v) = 0$. 
    \item If $u$ and $v$ are honest, then  $uv \in E$ if and only if $c_u(v) =1$. 
            \item For each $v \in V$,   $|\comm(v)| = \cO(\log(n))$. 
        \item For each \(v\in V\), \(|\comm^{-1}(v)|= \cO(\log^2n)\).   
\end{enumerate}

In the following, we say that the committees of \(u\) and \(v\) \emph{agree} in the edge \(uv\) if \(c_u(v)= c_v(u) \). In our definition, we allow for the information held by the honest nodes in the committee of $u$ regarding the neighborhood of $N(u)$ to be incorrect. In other words, the committee of a node may believe that there are \emph{false edges} (not present in $G$) as well as \emph{missing edges}, i.e.,  believe that there is no edge between a pair of nodes when in fact there is one. Our definition also allows that a pair of committees agree in a false edge. Nevertheless, condition 3 states that the committees of an honest node \(u\) never agree   with the committee of a Byzantine node \(v\) on a false edge \(uv\).  Moreover, the committee structure have full knowledge of the edges between honest nodes in the embedded graph $G$. In particular, the committees of a pair of  honest vertices linked by an edge agree on the existence of that edge. 

In \cite{augustine2022byzantine} it is given an algorithm in the Byzantine congested clique model, which with high probability computes a committee structure of the input graph in a poly-logarithmic number of rounds.

\begin{proposition}\cite{augustine2022byzantine}\label{TEO-Comm} 
There exists a randomized protocol that, given a graph $G=(V,E)$ embedded in an $n$-node Byzantine congested clique model and an honest node \(\ell\),   w.h.p. computes a committee structure \(\comm(G)\) of \(G\) in $\mathcal{O}(\polylog(n))$ rounds.
\end{proposition}

Given the preceding result, the necessity of the aforementioned leaders’ committee $L$ is evident: in order to obtain a committee structure, the existence of an honest node is required, which—due to the characteristics of the Byzantine Congested Clique Model—cannot be identified with absolute certainty. To address this issue, for each leader $l$ in $L$ we execute the corresponding subroutines, including the one responsible for generating a committee structure \(\comm(G)\) of \(G\). By leveraging the properties of the leaders’ committee, we obtain with high probability that at least two-thirds of the executions are reliable.

\subsection{Committee-Based Communication}

The committee structure allows us to simulate honest communications between any pair of nodes, even if one of them is Byzantine. In fact, the complications that Byzantine nodes may cause are reduced to the set of edges that communicate with their committee. From now on, when a node $u$ wants to send a message to a node $v$, we will have the nodes of \(\comm(u)\) communicate with the nodes of \(\comm(v)\). Since the majority of the nodes in a committee are honest, the communication is reliable. In this way, we simulate the rounds of an algorithm in the congested clique model (without Byzantine nodes) through the committees. Obviously, this will have consequences on the number of rounds. Indeed, an algorithm of $R$ rounds in the congested clique model can be simulated by \(\cO(R\cdot\polylog(n))\) rounds through the committees.

Let us be more precise and define the problem \(\textsc{Committe Relay}\). In the input there is a committee structure of the input graph. For each node \(u\in V\) the honest nodes \(w\in \comm(u)\) have for each  node 
$v \in V$ a message \(m(u,v)\) addressed to node \(v\in V\) (i.e. there are \(n\) messages per node, \(n(n-1)\) messages in total). The task consists in ensuring that for each \(v\in V\), each honest node \(w \in \comm(v)\) learns the messages \(m(u,v)\), for each \(u\in V\). 

Augustine et al. provide in \cite{augustine2022byzantine} (Lemma 13 and 14,  Algorithms 3 and 4) an algorithm in the Byzantine congested clique model that, assuming a Committee structure,  solves the \(\textsc{Committe Relay}\) problem in a polylogarithmic number of rounds.

\begin{proposition}\cite{augustine2022byzantine}
\label{prop:comm} Given a committee structure \(\comm(G)\), there is an algorithm in the Byzantine congested clique model solving \(\textsc{Committee Relay}\) w.h.p. in  $\mathcal{O}(\polylog(n))$ rounds of communication. 
\end{proposition}

\subsection{Graph Reconstruction}                                                            
Graph reconstruction is a communication task aimed at enabling each node to learn all the edges of the input graph $G$. Obviously, if the input graph is arbitrary, this task requires communicating many bits. For this reason, the reconstruction task is focused on particular classes of graphs.

Given a class of graphs \(\mathcal{G}\), the problem \(\reco{\mathcal{G}}\) is the following distributed task:

\medskip
\noindent{\bf Input:} A graph \(G = (V,E)\)\\
\noindent{\bf Task:}
\begin{itemize}
\item If \(G\in \mathcal{G}\), then every node outputs \(E\).
\item If \(G\notin \mathcal{G}\), then every node \emph{rejects}. 
\end{itemize}

If we denote by \(\mathcal{G}_n\) the subset of \(\mathcal{G}\) of \(n\)-node graphs, it is easy to see that any algorithm solving \(\reco{\mathcal{G}}\) in the congested clique model requires \(\Omega( \log(|\mathcal{G}_n|)/(n\log n))\) rounds of communication. Montealegre et al. show in  \cite{montealegre2020graph} that the trivial lower-bound is tight (up to logarithmic factors), even for deterministic algorithms. More precisely, the authors of \cite{montealegre2020graph} show the following proposition. 

\begin{proposition}{~\cite{montealegre2020graph}} Let $\mathcal{G}$ be an arbitrary class of graphs. There exists a deterministic algorithm in the congested clique model solving \(\reco{\mathcal{G}}\) in $\mathcal{O}\left(\log(|\mathcal{G}_n|)/(n\log n))\right)$ rounds. 
\label{prop:recon}
\end{proposition}

\section{Recognizing Hereditary Graph Classes}

In the following, we fix an arbitrary hereditary graph class \(\mathcal{G}\), and consider \(G = (V,E)\) an instance of \(\byzantinereco{\mathcal{G}}(|B|)\). Let us assume that we have computed a committee structure \(\comm(G) =\{\comm(u)\}_{u\in V}\) of \(G\) according to Proposition~\ref{TEO-Comm}.


\subsection{The Agreement and Disagreement Edges}

The committee structure \(\comm(G)\) induces two sets of edges over \(V\). The set of \emph{agreement edges} of $G$ is the set of edges for which the committees of its endpoints agreed about its existence. More formally,
$$\mathcal{A} = \left\{ uv \in \binom{[n]}{2} \:\:|\:\: c_{u}(v) = c_{v}(u) = 1\right\}$$ 
The set of agreement edges is not necessarily a subset of \(E\). Indeed, two non-adjacent Byzantine nodes could trick their committees and make them believe that there is an edge between them. 
Similarly, we define the set of \emph{disagreement edges} of $G$ as the edges for which there is no agreement on the existence of the edge between the two committees. Formally, 
$$\mathcal{D}= \left\{ uv \in \binom{[n]}{2} \:|\: c_u(v) \neq c_v(u)\right\}.$$
Observe that at least one endpoint of a disagreement edge must be Byzantine. 

Then, the \emph{agreement graph} of $G$ and \(\comm(G)\) is the graph given by:
$$G_{\mathcal{A}} = (V,\mathcal{A});$$
and the \emph{disagreement graph} of \(G\) and \(\comm(G)\) is the graph:
$$H_{\mathcal{D}} = (V,\mathcal{D}).$$


While the nodes have no idea which agreement edges are real, they have clear suspicion about the disagreement edges. 


Our algorithm consists in two parts. In the first part, we communicate to every honest node the set of agreement and disagreement edges, namely \(\mathcal{A}(G)\) and \(\mathcal{D}(G)\). Then, each node locally computes the distance of \(\mathcal{G}\) to the hereditary class via a procedure based in brute force.


\subsection{\texorpdfstring{Measuring Gap to $\mathcal{G}$}{Measuring Gap to G}}

If an honest node has knowledge of both the agreed and the disagreement graphs $G_{\mathcal{A}}$ and $G_{\mathcal{D}}$, it suffices for it to perform the following procedure. First, it verifies whether the agreement graph $G_{\mathcal{A}}$ belongs to the hereditary class \(\mathcal{G}\); if so, it accepts. Otherwise, it selects a subset of nodes $S$ (which we assume are Byzantine nodes) such that:

\begin{itemize}
\item $|S| = |B|$
\item Every edge in $\mathcal{D}$ has at least one endpoint incident to $S$.
\end{itemize}

For this subset, it tests different combinations of incident edges, in other words, try any subset $F \subseteq S \times V$ and tests if $F \cup (G_{\mathcal{A}} \backslash (S \times V)) = G^F$ lies in the class, if $G^F\in \mathcal{G}$, it accepts; otherwise, it proceeds to another combination of edges $F$, and so on until all possibilities are exhausted. If none leads to acceptance, it selects another set of suspicious nodes $S$ and repeats the process by checking all possible incident-edge combinations for that set. In other words, this procedure assumes that all suspicious nodes are Byzantine and systematically verifies all possible misreports they could have made, accepting if a configuration $G^F$ belonging to the class is found, and rejecting if none exists after exhausting all possibilities.

The second property imposed on the set \(S\) is referred to as the \emph{consistency property}. This property essentially prevents the selection of an impossible suspicious set, namely one that would imply the existence of more than \(|B|\) Byzantine nodes. Indeed, if the chosen set \(S\) leaves some disagreement edges uncovered, then by the properties of a committee structure we know that at least one endpoint of each such edge must be Byzantine. Hence, the selected set \(S\) would necessarily assume the existence of at least \(|B|+1\) Byzantine nodes.

We deduce the following lemma:

\begin{lemma}\label{lem:measuringgap}
Given a committee structure \(\comm(G)\) of a graph \(G\), there exists a locally computable algorithm such that, if a node has knowledge of the sets \(\mathcal{A}\) and \(\mathcal{D}\), it is able to perform the following:
\begin{itemize}
    \item If \(G \in \mathcal{G}\), then the node outputs \emph{accept}.
    \item If \(G\) is \(|B|\)-far from \(\mathcal{G}\), then the node outputs \emph{reject}.
\end{itemize}
\end{lemma}

\begin{proof}

It is clear that if the input graph $G$ belongs to the class \(\mathcal{G}\), then there exists the set $S = B$ of nodes and the subset $F = (S \times V) \cap E$ for which is clear $G^F = G$, and thus the algorithm accepts. By the other hand, assume, for the sake of contradiction, that $G$ is $|B|$-far from $\mathcal{G}$ and that there exists a subset $S \subset V$ satisfying the aforementioned properties, together with a set $F$ for this subset $S$ such that $G^F \in \mathcal{G}$.  

Fixing these sets $S$ and $F$, consider the family of graphs $G_i'$ such that for all $i \in [|B|+1]$ we have $G_i' \in \mathcal{G}$, but $G_i = G[V(G_i')] \notin \mathcal{G}$. Such a family exists because $G$ is $|B|$-far from $\mathcal{G}$ while $G^F \in \mathcal{G}$.  

Now observe that there exists some $k \in [|B|+1]$ such that $S \cap V(G_k') = \emptyset$. This follows directly from the fact that $|S| = |B|$. Consequently, there must exist a vertex $v \in V(G_k')$ such that $v \in B$. Otherwise, every edge $e \in E(G_k')$ would have honest endpoints, and since $S \cap V(G_k') = \emptyset$, by Condition~4 of a committee structure we would obtain $G_k' = G_k$, which is impossible because $G_k' \in \mathcal{G}$ while $G_k \notin \mathcal{G}$.  
Moreover, it must hold that $d_{G_{\mathcal{D}}}(v) \geq 1$, for otherwise, by the converse of Condition~3 of a committee structure, we would again obtain $G_k' = G_k$, contradicting the assumption as before. Finally, since $d_{G_{\mathcal{D}}}(v) \geq 1$, we conclude that $S$ does not satisfy the consistency property. This yields a contradiction, thereby proving that there do not exist $S$ and $F$ such that $G^F \in \mathcal{G}$.
\end{proof}

\begin{figure}[H]
    \centering
        \includegraphics[width=0.3\linewidth]{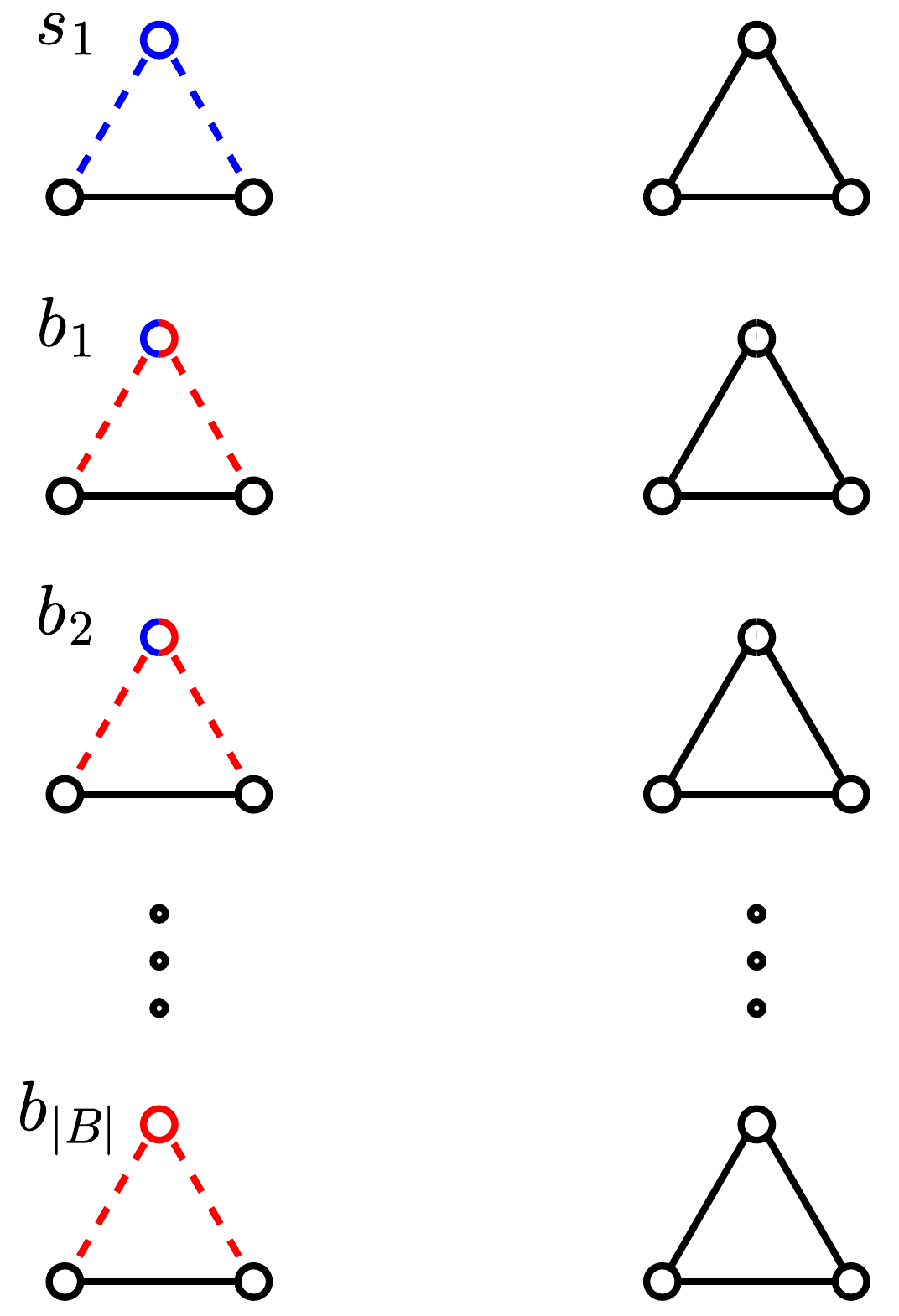}
        \caption{\centering Example of the graph $G$ and the possible graph $G^F$ obtained that belongs to the class of forests.}
        \label{fig:prooflemma1}
\end{figure}

At first glance, it is not clear why the consistency property is needed in $S$, but it is necessary to ensure the correct functioning of the procedure. Indeed, let us consider that the consistency property is not required to select a set $S$. Then, take $G = \cup_{i=1}^{|B|+1} T_i \quad \text{where} \quad T_i = (\{v_1^i, v_2^i, v_3^i\}, \{\{v_1^i,v_2^i\}, \{v_2^i, v_3^i\}, \{v_3^i,v_1^i\}\})$. Clearly $G \in \mathcal{T}$, the hereditary class of forests. Next, take $B = \{ v_2^i\}_{i>1}$ and $\mathcal{D} = \{\{v_1^i,v_2^i\}, \{v_2^i, v_3^i\}\}_{i>1}$. Then, by taking $S=\{v_2^i\}_{i<|B|+1}$ and $F=\{\emptyset\}$, we obtain $G^F = \cup_{i=1}^{|B|+1} \big(\{v_1^i, v_2^i, v_3^i\}, \{ \{v_3^i, v_1^i\}\}\big)$ which is clearly in $\mathcal{T}$ (See Figure \ref{fig:prooflemma1}). Thus, the procedure would accept for this set $S$ and $F$. However, $G$ is $|B|+1$-far from $\mathcal{T}$, failing to solve the promise problem. Thus showing the importance of using the information contained in the disagreement graph.

\subsection{Reconstruction of the Agreement and Disagreement Graphs}

The Measuring Gap to $\mathcal{G}$ introduced in the previous subsection effectively differentiates between the \emph{yes} and \emph{no} instances of the given promise problem. Our current objective is to equip the honest nodes with sufficient information to execute it. Note that by the given Lemma~\ref{lem:measuringgap} knowing both the agreement edges \(\mathcal{A}\)
and the disagreement edges \(\mathcal{D}\) is sufficient for this computation. Consequently, our primary goal becomes ensuring that each node receives all information about both types of edges. This dissemination of information can be efficiently executed using reconstruction algorithms applied to both the agreement and disagreement graphs. Since these graphs exhibit distinct characteristics, they require separate approaches in their treatment.

\subsubsection{Reconstruction of the Disagreement Graph} Observe that there are no disagreement edges between honest nodes (condition 4 of the definition of the committee structure). Therefore, the honest nodes have at most \(|B|\) neighbors in the disagreement graph. In other words, any node of degree greater than \(|B|\) in the disagreement graph is necessarily Byzantine. Let us call \(d_{\mathcal{D}}(u)\) the degree of \(u\) in the disagreement graph \(H_\mathcal{D}\). We partition \(V\) into two sets, the set \(V^+\) containing the nodes \(u\) such that \(d_{\mathcal{D}(u)} > |B|\), and  \(V^- = V\setminus V^+\). 

We now describe the reconstruction algorithm. 
At the beginning, the nodes in \(\comm(u)\) broadcast \(d_{\mathcal{D}}(u)\) to every other node, for each \(u\in V\). By a majority vote over \(\comm(u)\), every node learns \(d_{\mathcal{D}}(u)\). This procedure takes one round. 

Then, for each \(u \in V^-\) the nodes in \(\comm(u)\) simply broadcast the list of edges incident to \(u\). By the definition of \(V^-\), this list contains at most \(|B|\) edges. Again, by a majority vote over \(\comm(u)\), every honest node learns the disagreement edges incident to \(u\). Moreover, every node also learns all the disagreement edges with one endpoint in \(V^-\) and the other in \(V^+\). It remains to communicate the edges with both endpoints in \(V^+\). For each \(u \in V^+\), the nodes in \(\comm(u)\) broadcast the list of disagreement edges incident to \(u\) with the other endpoint also in \(V^+\). Since the nodes in \(V^+\) are Byzantine, \(u\) has at most \(|B|\) such edges. Once again, by majority vote over the respective committees, every node learns the disagreement edges with both endpoints in \(V^+\).

By definition of the committee structure (condition 6), each node \(v\) participates in \(\cO(\log^2 n)\) committees (i.e. \(|\comm^{-1}(v)| = \cO(\log^2)\)). For each \(u\in \comm^{-1}(v)\) node \(v\) has to communicate at most \(|B|\) edges, using at most \(|B|\) rounds. Then, the whole reconstruction algorithm takes \(\cO(|B|\log^2(n))\) rounds.  We deduce the following lemma.

\begin{lemma}\label{lem:disagrementreco} Given a committee structure \(\comm(G)\) of a graph \(G\), there is an algorithm in the Byzantine congested clique model in which every honest node outputs the set of all disagreement edges \(\mathcal{D}\). The algorithm runs in
$$\cO\left( |B|\log^2(n)\right)\textrm{ rounds.}$$ 
\end{lemma}

\subsubsection{Reconstruction of the Agreement Graph} The reconstruction of the agreement graph is simpler, as we use the reconstruction algorithm of Montealegre et al. given in Proposition~\ref{prop:recon}. We denote by \(\mathcal{G}^{|B|}\)  the class of all possible agreement graphs obtained from a graph \(G\in \mathcal{G}\) with \(|B|\) Byzantine nodes. 

For each  \(n\)-node graph \(G\in \mathcal{G}_n\), the adversary has   \(\binom{n}{|B|}\) possible choices of where to place the Byzantine nodes. Each Byzantine node can create \((n-|B|)\) disagreement edges with honest nodes, and an arbitrary subgraph with the other \(|B|-1\) Byzantine nodes. 
Therefore, the cardinality of \(\mathcal{G}^{|B|}_n\) is

$$| \mathcal{G}^{|B|}_n| \leq  |\mathcal{G}_n|\cdot \binom{n}{|B|} \cdot 2^{|B|(n-|B|)} \cdot 2^{\binom{|B|}{2}}\leq  |\mathcal{G}_n|\cdot \binom{n}{|B|} \cdot 2^{|B|n}. $$
Then, Proposition~\ref{prop:recon} provides an algorithm \(\textsf{Algo}\) for \(\reco{\mathcal{G}^{|B|}}\) in the Congested Clique model running in 

$$\cO\left( \frac{\log|\mathcal{G}_n|}{n\log n} + \frac{|B|}{\log n}\right)\textrm{ rounds.}$$

Observe that when the input graph \(G\) belongs to \(\mathcal{G}\) the output of the nodes in \(\textsf{Algo}\) is the set of edges of some realization of the agreement graph. Conversely, if the nodes output \emph{reject} we necessarily have that \(G\notin \mathcal{G}\). When the input graph is not in \(\mathcal{G}\), the output of \(\textsf{Algo}\) is  either the set of edges of an agreement graph of \(G\), or \emph{reject}.

Using Proposition~\ref{prop:comm}, we simulate \(\textsf{Algo}\) through the committee structure. In this way we obtain an algorithm in the Byzantine congested clique model, where every honest node outputs the value of  \(\reco{\mathcal{G}^{|B|}}\). More precisely, either every honest node recovers all the edges in the agreement graph of the input graph, or every honest node \emph{rejects}. We deduce the following lemma.

\begin{lemma}\label{lem:agrementreco} Given a committee structure \(\comm(G)\) of a graph \(G\), there is an algorithm in the Byzantine Congested Clique model in which every honest node either outputs the set of all agreement edges \(\mathcal{A}\), or \emph{rejects}. When the output of the honest nodes is \emph{reject}, the input graph \(G\) is not in \(\mathcal{G}\). The algorithm runs in
$$\cO\left(\left( \frac{\log|\mathcal{G}_n|}{n} + |B|\right)\cdot \polylog(n)\right)\textrm{ rounds.}$$ 
\end{lemma}

\subsection{Bringing It All Together}

We now have all the ingredients for proving Theorem~\ref{theo:main}. 

\main*

\begin{proof}
An algorithm for  \(\byzantinereco{\mathcal{G}(|B|)}\) consists in the following steps, described also in Algorithm~\ref{algo:main}. First, we compute a committee of leaders \(L\) using the algorithm given by Proposition~\ref{L1}. 

For each \(\ell \in L\), we execute the algorithm of Proposition~\ref{TEO-Comm}  for building a committee structure of the input graph \(G\) associated to \(\ell\) (i.e. we form \(|L| = \cO(\log n)\) committee structures of \(G\), one for each leader \(\ell \in L\). Each committee structure defines a set of agreement and disagreement edges. Then, the nodes try to reconstruct these sets of edges using the algorithms of Lemma~\ref{lem:disagrementreco} and \ref{lem:agrementreco}. 

If an honest node cannot reconstruct the agreement graph, it sends \emph{reject} to leader \(\ell\). Otherwise, every honest node knows all the edges in the agreement and disagreement graphs. Using that information the nodes can locally compute the Measuring Gap to $\mathcal{G}$ procedure, finally the nodes communicate \emph{accept} to the leader \(\ell\) if the procedure end in acceptance. Otherwise, the nodes communicate \emph{reject} to \(\ell\).

Each leader \(\ell \in L\) then receives the decision of the honest nodes, and adopts the majority decision. Then it communicates that decision to every other node of the graph. Finally, the nodes output the decision of the majority of the leaders. 

 \begin{algorithm}[h]
	\caption{Algorithm for \(\byzantinereco{\mathcal{G}(|B|)}\) from the perspective of an honest node \(v\)}
 \label{algo:main}
    \begin{flushleft}
        \textbf{Require:} An input graph $G = (V,E)$\\
        \textbf{Ensure:} When \(G\in \mathcal{G}\) every honest node accepts. When \(G\) is \(|B|\)-far from \(\mathcal{G}\) every hones node rejects. 
    \end{flushleft}
	\begin{algorithmic}[1]
    \State Invoke the algorithm from Proposition~\ref{L1} and obtain a committee of leaders \(L\). \newline/* \(v\) learns the identities of all vertices in \(L\). */
    \For{ $\ell \in L$}
 
    \State Invoke the algorithm from Proposition~\ref{TEO-Comm} and obtain \(\comm(G)\) of \(G\). \newline
     /*  \(v\) knows the identities of all the members of \(\comm(u)\) for every \(u\in V\). Node \(v\) also knows all the consistency bits, the agreement and disagreement edges incident to the nodes \(u\in \comm^{-1}(v)\). */

        \State Invoke the algorithm from Lemma~\ref{lem:disagrementreco}  to reconstruct $\mathcal{D}$. \newline
        /*  \(v\) learns all the edges in \(\mathcal{D}\).*/ 
        
        \State Invoke the algorithm from Lemma~\ref{lem:agrementreco}  to try to reconstruct $\mathcal{A}$. 
         
        \If { the output of the algorithm of  Lemma~\ref{lem:agrementreco} is \emph{reject}}
        \State Send \emph{reject} to $\ell$. 
        \Else ~/*\(v\) learns all the edges in \(\mathcal{A}\)*/
        \State \(v\) locally computes the Measuring Gap to \(\mathcal{G}\) procedure using \(\mathcal{A}\) and \(\mathcal{D}\).
        \If {the Measuring Gap to \(\mathcal{G}\) procedure accepts}
        \State Send \emph{accept} to $\ell$. 
        \Else
        \State Send \emph{reject} to $\ell$. 
        \EndIf
        \EndIf
    \EndFor  
    \If{ $v\in L$}
    \State Pick the most repeated decision communicated to \(v\) as a leader
    \State Communicate that decision to every other node in \(V\). 
    \EndIf
    \State Output the decision communicated by the majority of the nodes in \(L\). 
    \end{algorithmic} 
    \end{algorithm}

The correctness of the subroutines are given by the correctness of the algorithms of Proposition~\ref{L1} and \ref{TEO-Comm}, that obtain a committee of leaders and committee structures with high probability. Lemma~\ref{lem:disagrementreco} and \ref{lem:agrementreco} also provide the correctness of the reconstruction algorithms for the agreement and disagreement graphs. Observe that when the agreement graph is not reconstructed, necessarily \(G\notin \mathcal{G}\), and then the honest nodes should reject. When the agreement graph can be successfully reconstructed, we have by Lemma~\ref{lem:measuringgap} that the result of the Measuring Gap to $\mathcal{G}$ procedure allows the honest nodes to distinguish the instances where \(G\in \mathcal{G}\) from the instances \(|B|\)-far from \(\mathcal{G}\). We deduce that Algorithm~\ref{algo:main} successfully solves \(\byzantinereco{\mathcal{G}(|B|)}\).

With respect to the complexity, the number of rounds is dominated by the subroutine of Lemma~\ref{lem:agrementreco}, used to try to reconstruct the agreement graph. This algorithm runs in

$$\cO\left(\left( \frac{\log|\mathcal{G}_n|}{n} + |B|\right)\cdot \polylog(n)\right)\textrm{ rounds.}$$ 

Since we have to run this algorithm for each \(\ell \in L\) and \(|L|=\cO(\log n)\), we obtain that the algorithm runs in

$$\cO\left(\left( \frac{\log|\mathcal{G}_n|}{n} + |B|\right)\cdot \polylog(n)\right) \textrm{  rounds in total.}$$
\end{proof}

\section{Impossibility of Forest Recognition }

In this section we prove Theorem~\ref{theo:impossibility}.


\secondmain*

\begin{proof}

Suppose that there exists an algorithm $\textsf{ALG}$   for solving {\sc Byzantine-Recognition-of-Forests}\((f)\) for some 
\(f< |B|\). We are going to contradict the correctness of \(\textsf{ALG}\) by defining two instances, namely \(G_{Yes}\) and \(G_{No}\), such that \(G_{\text{Yes}}\) is a forest while \(G_{No}\) is \(f\)-far from being a forest. In each of these instances, we will place the Byzantine nodes and ask them to lie in a specific way, so that from the perspective of the honest nodes, the instances are indistinguishable.

A \emph{triangle} \(T\) is a graph of three vertices, containing all the edges between these vertices: $T = \{v_1, v_2, v_3\}, \{\{v_1,v_2\}, \{v_1, v_3\}, \{v_2,v_3\}\}$.

The graph \(G_{No}\) is simply defined as the disjoint union of \(f+1\) triangles, denoted \(T_1,\dots, T_{f+1}\).  Observe that \(G_{\text{No}}\) contains \(f+1\) disjoint cycles, hence it is \(f\)-far from the class of forests. Let us denote \(\{v_1^i, v_2^i, v_3^i\}\) the vertex set of \(T_i\). The graph \(G_{\text{Yes}}\) is obtained from \(G_{\text{No}}\) by removing edge \(\{v_1^i, v_2^i\}\) from each triangle \(T_i\) (See Figure \ref{fig:proofteo2a}). Observe that each component of \(G_{\text{Yes}}\) is a three node path. Therefore, \(G_{\text{Yes}}\) belongs to the class of forests. 

For each \(k \in \{1,2\}\), let us denote \(G^k_{\text{Yes}}\) (respectively \(G^k_{\text{No}}\)), the instance of {\sc Byzantine-Recognition-of-Forests}\((f)\) when the input graph is \(G_{\text{Yes}}\) (respectively \(G_{\text{No}}\)) and a set of \(f+1 \leq |B|\) Byzantine nodes are located in vertices \(v^i_k\), for each \(i\in \{1, \dots, f+1\}\). 

In the instances \(G_{\text{Yes}}^k\), the Byzantine nodes pretend as if the edge \(\{v^i_1,v^i_2\}\) exists, and apart from that, they behave just like an honest node. In the instances \(G_{\text{No}}^k\), the Byzantine nodes pretend as if the edge \(\{v^i_1,v^i_2\}\) does not exist, and apart from that, they behave just like an honest node.

We obtain that  \(G_{\text{Yes}}^1\) is indistinguishable form \(G_{\text{No}}^2\) (See Figure \ref{fig:proofteo2b}): for each \(i \in \{1,\dots, f+1\}\) the view of \(v^i_3\) is the same on both graphs, node \(v_1^i\) claims that the edge \(\{v^i_1, v^i_2\}\) does exist and node \(v_2^i\) claims that the edge \(\{v^i_1, v^i_2\}\) does not exist. Since the nodes of each graph execute the same algorithm under the same local input, the global algorithm $\textsf{ALG}$ executed by the two instances are statistically identical. However, $\textsf{ALG}$ must be able to distinguish between \(G_{\text{Yes}}^1\) and \(G_{\text{No}}^2\) with probability greater than \(1/2\), which is a contradiction. 
\end{proof}

\begin{figure}[H]
    \centering
    \begin{subfigure}{0.3\textwidth}
        \centering
        \includegraphics[width=\linewidth]{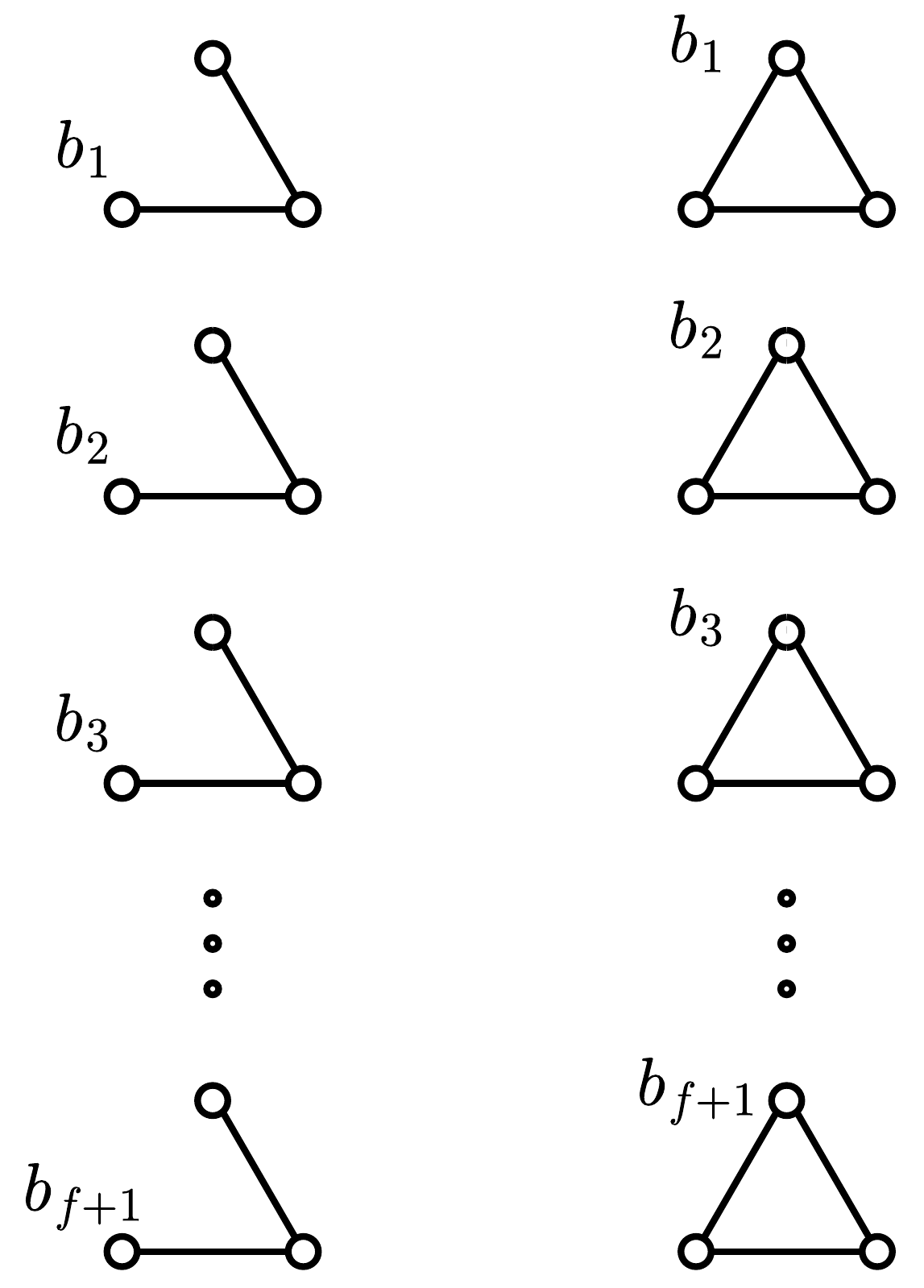}
        \caption{\centering \(G_{\text{Yes}}\) and \(G_{\text{No}}\)}
        \label{fig:proofteo2a}
    \end{subfigure}
    \hspace{2cm}
    \begin{subfigure}{0.3\textwidth}
        \centering
        \includegraphics[width=\linewidth]{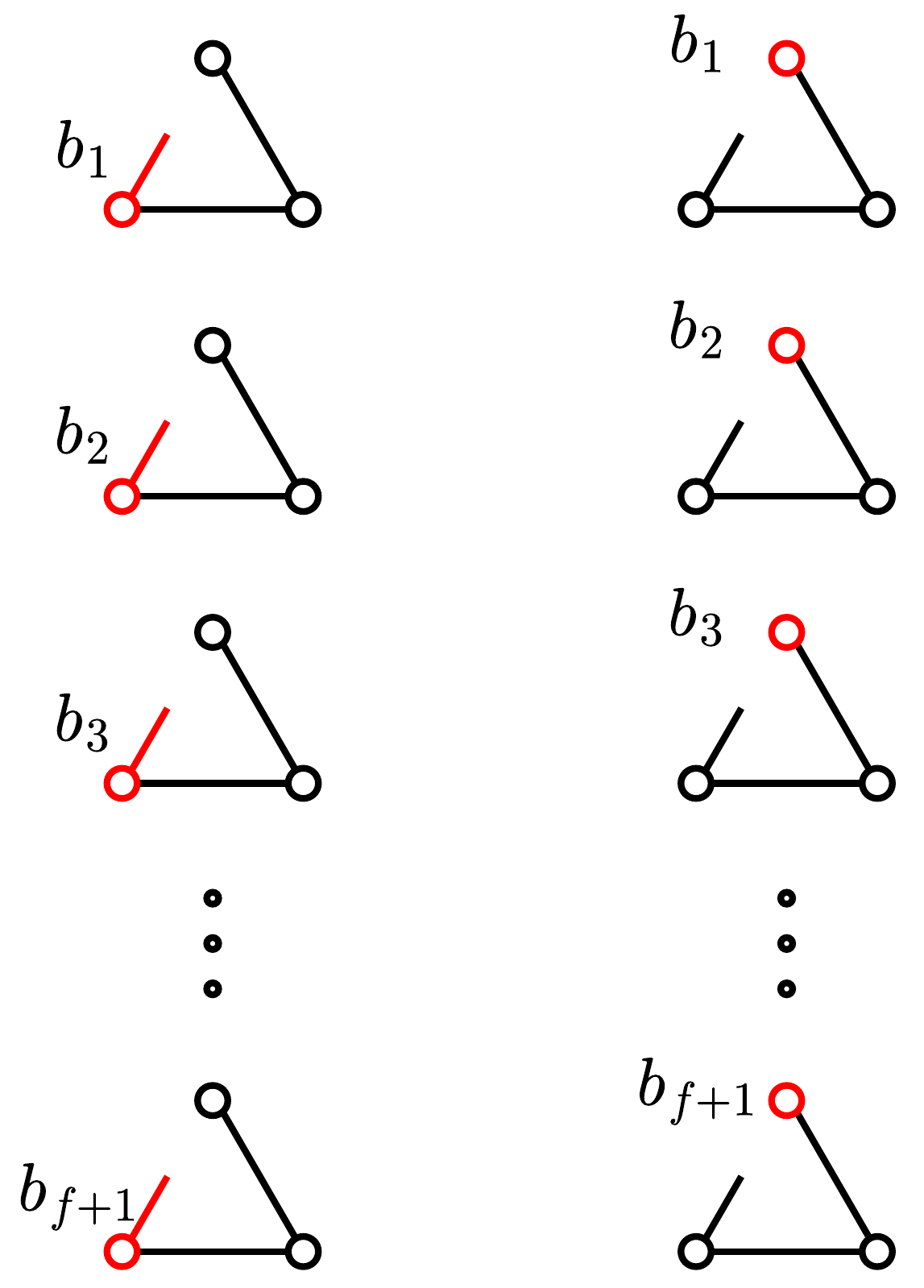}
        \caption{\centering \(G_{\text{Yes}}^1\) and \(G_{\text{No}}^2\)}
        \label{fig:proofteo2b}
    \end{subfigure}
\caption{\centering Examples of the input graphs \(G_{\text{Yes}}\) and \(G_{\text{No}}\), and representation of the indistinguishable instances obtained from them }
    \label{fig:graphs}
\end{figure}

\section{Discussion} 

In this work we have presented an algorithm in the Byzantine congested clique model that recognizes any hereditary graph class \(\mathcal{G}\). We have shown that, with high probability, our algorithm accepts if \(G\) belongs to \(\mathcal{G}\) and rejects if \(G\) is \(|B|\)-far from \(\mathcal{G}\).

Our algorithm highlights the robustness of hereditary graph classes in the Byzanyine framework,  suggesting an inherent resistance to malicious manipulation in distributed systems. 

We have also shown that the gap of our algorithm is the best possible, in the sense that there are hereditary graph classes that cannot be recognized with a gap smaller than \(|B|\).
We believe that the impossibility bound we have established could be more general, being applicable to every non-trivial hereditary graph class. 

Furthermore, we suggest to study how other problems behave in this model. In particular, optimization and approximation problems in the presence of Byzantine nodes present a challenging and relevant field of study. The complexity introduced by the presence of Byzantine nodes in optimization problems could reveal new strategies for mitigating their influence or even utilizing their presence to improve system efficiency.

\bibliography{biblio}

\end{document}